\documentclass[11pt]{article}
\usepackage{url}
\usepackage{graphicx}
\usepackage[linktocpage=true,pagebackref=true]{hyperref}

\usepackage{enumerate}

\usepackage{fullpage}
\usepackage{amsmath}
\usepackage{amssymb}
\usepackage{amsthm}

\usepackage{algorithmic}
\usepackage{algorithm}
\floatname{algorithm}{Procedure}

\newtheorem{theorem}{Theorem}
\newtheorem{lemma}[theorem]{Lemma}

\newtheorem{corollary}[theorem]{Corollary}
\newtheorem{property}[theorem]{Property}

\newcommand{\comment}[1]{}

\newcommand{\calO}{\mathcal{O}}
\newcommand{\calK}{\mathcal{K}\,}
\newcommand{\calP}{\mathcal{P}}
\newcommand{\calR}{\mathcal{R}\,}
\newcommand{\calL}{\mathcal{L}\,}
\newcommand{\calD}{\mathcal{D}}

\renewcommand{\setminus}{-}

\newcommand{\squishlist}{
 \begin{list}{$\bullet$}
  { \setlength{\itemsep}{0pt}
     \setlength{\parsep}{3pt}
     \setlength{\topsep}{3pt}
     \setlength{\partopsep}{0pt}
     \setlength{\leftmargin}{1.5em}
     \setlength{\labelwidth}{1em}
     \setlength{\labelsep}{0.5em} } }

\newcommand{\squishlisttwo}{
 \begin{list}{$\bullet$}
  { \setlength{\itemsep}{0pt}
     \setlength{\parsep}{0pt}
    \setlength{\topsep}{0pt}
    \setlength{\partopsep}{0pt}
    \setlength{\leftmargin}{2em}
    \setlength{\labelwidth}{1.5em}
    \setlength{\labelsep}{0.5em} } }

\newcommand{\squishend}{
  \end{list}  }

\begin{document}

\renewcommand{\thefootnote}{\fnsymbol{footnote}}




\title{Routing regardless of Network Stability}
\author{
  Bundit Laekhanukit\thanks{
      School of Computer Science,
      McGill University.
      Supported by a Dr. and Mrs. Milton Leong Fellowship and by NSERC grant 288334.
      Email: {\tt blaekh@cs.mcgill.ca}
  } \and
  Adrian Vetta\thanks{
      Department of Mathematics and Statistics and
      School of Computer Science,
      McGill University.
      Supported in part by NSERC grants 288334 and 429598.
      Email: {\tt vetta@math.mcgill.ca} 
  } \and
  Gordon Wilfong\thanks{
      Bell Laboratories.
      Email: {\tt gtw@research.bell-labs.com}
  }
}




\date{\today}

\maketitle



\renewcommand{\thefootnote}{\arabic{footnote}}
\begin{abstract}
How effective are interdomain routing protocols, such as the {\em Border Gateway Protocol}, 
at routing packets?
 Theoretical analyses have attempted to answer this 
 question by ignoring the packets and instead 
 focusing upon protocol stability. To study stability, it suffices to model only the control plane (which determines
the routing graph) -- an
 approach taken in the {\em Stable Paths Problem}.
  To analyse packet routing requires modelling the interactions between the
  control plane and the forwarding plane (which determines where packets are forwarded), and our first contribution is to introduce such 
  a model. We then examine the effectiveness of packet routing in this model for the broad class
  {\em next-hop preferences with filtering}. 
  Here each node $v$ has a {\em filtering list} $\mathcal{D}(v)$ consisting of nodes it does not want its packets 
  to route through. Acceptable paths (those that avoid nodes in the filtering list) are ranked according to
  the {\em next-hop}, that is, the neighbour of $v$ that the path begins with.
  On the negative side, we present a strong inapproximability result.
  For filtering lists of cardinality at most one, given a network in
which an equilibrium is guaranteed to exist, it is NP-hard to approximate the
maximum number of packets that can be routed to within a factor of $n^{1-\epsilon}$, for any constant $\epsilon >0$.
  On the positive side, we give algorithms to show that in two fundamental cases {\em every} packet will eventually route
with probability one. 
  The first case is when each node's filtering list contains only itself, that is, $\mathcal{D}(v)=\{v\}$; this is the fundamental case in which 
  a node does not want its packets to cycle.
  Moreover, with positive probability every packet will be routed before the control plane reaches an equilibrium.
  The second case is when all the filtering lists are empty, that is, $\mathcal{D}(v)=\emptyset$. Thus, with probability one
  packets will route even when the nodes do not care if their packets cycle!
  Furthermore, with probability one every packet will route even when the control plane has {\em no} equilibrium at all.
  To our knowledge, these are the first results to guarantee the possibility that all packets get routed without stability.
These positive results are tight -- for the general case of filtering lists of cardinality one, it is not possible to ensure 
that every packet will eventually route.
\end{abstract}



\section{Introduction}
\label{sec:intro}

 In the {\em Stable Paths Problem} (SPP) \cite{GSW02},
we are given a directed graph $G=(V,A)$ and a sink (or destination) node $r$.
Furthermore, each node $v$ has a ranked list of some of its paths to $r$. The lowest ranked entry in the list is the ``empty path''\footnote{Clearly, the empty path
is not a real path to the sink; we call it a path for clarity of exposition.}; paths that are not ranked are considered unsatisfactory.
This preference list is called $v$'s list of {\em acceptable paths}.
A set of paths, one path $\calP(v)$ from each node $v$'s list of
acceptable paths, is termed {\em stable} if 
\begin{itemize}
\item[(i)] they are {\em consistent}: if $u\in \calP(v)$, then $\calP(u)$ must be the subpath of $\calP(v)$ beginning at $u$, and 
\item[(ii)] they form an {\em equilibrium}: for each node $v$, $\calP(v)$ is the path ranked highest by $v$ of the form $v\calP(w)$ 
where $w$ is a neighbour of $v$. 
\end{itemize}
The stable paths problem asks whether a stable set of paths exists in the network.
The SPP has risen to prominence as it is viewed as a {\em static} description of the problem that 
 the Border Gateway Protocol (BGP) is trying {\em dynamically} to solve. 
BGP can be thought of as trying to find a set of stable routes to $r$ so that routers can use these routes to send packets to $r$.

Due to the importance of BGP, both practical and theoretical aspects of the SPP have been studied in great depth.
In the main text, to avoid overloading the reader with practical technicalities, we focus on the combinatorial aspects 
of packet routing; in the Appendix we discuss the technical aspects and present 
a motivating sample of the vast literature on BGP.
Two observations concerning the SPP, though, are pertinent here and motivate our work:

\begin{enumerate}[(1)]
\item Even if a stable solution exists, the routing tree induced by a consistent set of paths might not be spanning. 
Hence, a stable solution may {\bf not} actually correspond to a functioning network
-- there may be isolated nodes that cannot route packets to the sink!
Disconnectivities arise because nodes may prefer the empty-path to any of the paths proffered by its neighbours;
for example, a node might not trust certain nodes to handle its packets securely or in a timely fashion, 
so it may reject routes traversing such unreliable domains.
This problem of non-spanning routing trees has quite recently been studied in the context of a version of BGP called iBGP~\cite{VCVB12}.
In Section \ref{sec:general}, we show that non-connectivity is
a very serious problem (at least, from the theoretical side) by presenting an $n^{1-\epsilon}$
hardness result for the combinatorial problem of finding a {\em maximum cardinality stable subtree}. 

\item The SPP says nothing about the dynamic behaviour of BGP. 
Stable routings are significant for many practical
reasons (e.g., network operators want to know the routes
their packets are taking), but while BGP is operating at the {\em control plane} level,
packets are being sent at the {\em forwarding plane} level without waiting
for stability (if, indeed, stability is ever achieved). Thus, it is important to study network performance 
in the dynamic case. For example, what happens to the packets whilst a network is unstable?
This is the main focus of our paper: to investigate packet routing under network dynamics.
\end{enumerate}


Towards this goal, we define a distributed protocol, inspired by BGP, that
stops making changes to the routing graph (i.e., becomes stable) if it achieves a 
stable solution to the underlying instance of SPP.
The current routing graph itself is determined by the control plane but the movement of packets is determined by the 
forwarding plane. Thus, our distributed protocol provides a framework under which the control and forwarding planes interact; essentially, this primarily means that we need to understand the relative speeds at which
links change and packets move. 





%




Given this model, we analyse the resulting trajectory of packets. 
In
a stable solution, a node in the stable tree containing the sink would have its packets route
whereas an isolated node would not. 
For unstable networks, or for stable networks that have not converged, things are much more complicated.
Here the routes selected by nodes are changing over time and, as we shall see, this may cause the packets to cycle.
If packets can cycle, then keeping track of them is highly non-trivial.
Our main results, however, are that for two fundamental classes of preference functions 
(i.e., two ways of defining acceptable paths and their rankings)
all packets will route with probability one in our model. That is, there is an
execution of our distributed protocol
such that {\em every} packet in the network will reach the destination
(albeit, possibly, slowly) even in instances where the network has no stable solution. 
(Note that we are ignoring the fact that in BGP packets typically have a {\em time-to-live}
attribute meaning that after traversing a fixed number of nodes the packet
will be dropped.)  Furthermore,
when the network does have a stable solution, we are able to guarantee packet routing even before the 
time when the network converges. 

These positive results on the routing rate are to our knowledge,  the first results to guarantee 
the possibility of packet routing without stability.
The results are also tight in the sense that, for any more expressive class of preference function, 
our hardness results show that guaranteeing that all packets eventually route is not possible -- thus, 
packets must be lost.



\section{The Model and Results}\label{prelim}


We represent a network by a directed graph $G=(V,A)$ on $n$ nodes.
The destination node in the network is denoted by a distinguished node
$r$ called a {\em sink} node.
We assume that, for every node $v\in V$, there is at least one directed path in $G$ from $v$ to the sink $r$, 
and that the sink $r$ has no outgoing arc.
At any point in time $t$, each node $v$ chooses at most one of its out-neighbours $w$ as its {\em chosen next-hop};
 thus, $v$ selects one arc $(v,w)$ or selects none. 
These arcs form a {\em routing graph} $\calR_t$, each component of which 
is a {\em $1$-arborescence},
an {\em in-arborescence}\footnote{An {\em in-arborescence} is a
  graph $T$ such that the underlying undirected graph is a tree and
  every node has a unique path to a root node. } $T$ plus possibly one 
arc $(v,w)$ emanating from the root $v$ of $T$; for example, $T$ and
$T\cup\{(v,w)\}$ are both $1$-arborescences.
(If the root of a component does select a neighbour, then that
component contains a unique cycle.) 
When the context is clear, for clarity of exposition, we abuse the term {\em tree} to mean a $1$-arborescence, and we use the term {\em forest} to mean a set of trees.
A component (tree) in a routing graph is called a {\em sink-component} if it has the sink $r$ as a root;
all other components are called {\em non-sink components}. 

Each node selects its outgoing arc according to its preference list of acceptable paths.
We examine the case where these lists can be generated using two of the most 
common preference criteria in practice: {\em next-hop preferences} and {\em filtering}. 
For next-hop preferences, each node $v\in V$ has a ranking on its {\em out-neighbours}, nodes $w$ such that $(v,w)\in A$. 
We say that $w$ is the {\em $k$-th choice} of $v$ if $w$ is an out-neighbour of $v$ with the $k$-th rank.
For $k=1,2,\ldots,n$, we define a set of arcs $A_k$ to be such that $(v,w)\in A_k$ if $w$ is the $k$-th choice of $v$, i.e., $A_k$ is the set of {\em the $k$-th choice arcs}.
Thus, $A_1,A_2,\ldots,A_n$ partition the set of arcs $A$, i.e., $A=A_1\cup A_2\cup \ldots A_n$.
We call the entire graph $G=(V,A)$ an {\em all-choice} graph.
A {\em filtering list}, $\mathcal{D}(v)$, is a set of nodes that
$v$ {\bf never} wants its packets to route through.
We allow nodes to use filters and otherwise rank routes via next-hop preferences, namely 
{\em next-hop preferences with filtering}.
 
To be able to apply these preferences,
each node $v\in V$ is also associated with a path $\calP(v)$, called 
$v$'s {\em routing path}.
The routing path $\calP(v)$ may {\bf not} be the same as an actual $v,r$-path in the routing graph.
We say that a routing path $\calP(v)$ is {\em consistent} if $\calP(v)$ is a  $v,r$-path in the routing graph; otherwise, we say that $\calP(v)$ is {\em inconsistent}.
Similarly, we say that a node $v$ is {\em consistent} if its routing path $\calP(v)$ is consistent; otherwise, we say that $v$ is {\em inconsistent}.
A node $v$ is {\em clear} if the routing path $\calP(v)\neq\emptyset$, i.e., $v$ (believes it) has a path to the sink; otherwise, $v$ is {\em opaque}. 
%
We say that a node $w$ is {\em valid} for $v$ or is a 
{\em valid choice} for $v$ if $w$ is clear and $\calP(w)$ contains no nodes in the filtering list $\calD(w)$.
If $w$ is a valid choice for $v$, and $v$ prefers $w$ to all other valid
choices, then we say that $w$ is the {\em best valid choice} of
$v$.
A basic step in the dynamic behaviour of BGP is that, at any time $t$, some subset $V_t$ of nodes is {\em activated} meaning
that every node $v\in V_t$ chooses the highest ranked acceptable path 
$\calP(v)$ that
is consistent with one of its neighbours' routing paths at time $t-1$.  
The routing graph $\calR_t$ consists of 
the first arc in each routing path at time $t$. 

Protocol variations result from such things as restricting 
$V_t$ so that $|V_t|=1$, specifying the relative rates that nodes are chosen to
be activated and allowing other computations to occur between these
basic steps. In our protocol, we assume that activation orderings are {\em fair} 
in that each node activates exactly once in each time period -- a {\em round} -- the 
actual ordering however may differ in each round.
While our protocol is not intended to model exactly the behaviour of BGP, 
we tried to let BGP inspire our choices and to capture the
essential coordination problem that makes successful dynamic routing hard. 
Again, a detailed discussion on these issues and on the importance of a fairness-type criteria
is deferred to the Appendix.
\begin{algorithm}
\caption{Activate($v$)}
\label{algo:activate}
\begin{algorithmic}[1]
\REQUIRE A node $v\in V\setminus\{r\}$. 
\IF{$v$ has a valid choice}
   \STATE Choose the best valid choice $w$ of $v$.
   \STATE Change the outgoing arc of $v$ to $(v,w)$.
   \STATE Update $\calP(v) := v\calP(w)$ (the concatenation of $v$ and $\calP(w)$).
\ELSE
   \STATE Update $\calP(v) :=\emptyset$.
\ENDIF
\end{algorithmic}
\end{algorithm}

\begin{algorithm}
\caption{Protocol($G$,$r$,$\calR_0$)}
\label{algo:protocol}
\begin{algorithmic}[1]
\REQUIRE A network $G=(V,A)$, a sink node $r$ and a routing graph $\calR_0$
\STATE Initially, every node generates a packet. 
\FOR{round $t:=1$ to $\ldots$}
  \STATE Generate a permutation $\pi_t$ of nodes in $V\setminus\{r\}$ using an external algorithm~$\mathbb{A}$. 
  \STATE {\bf Control Plane:} Apply Activate$(v)$ to activate each node in the order in $\pi_t$. This forms a routing graph $\calR_t$. 
  \STATE {\bf Forwarding Plane:} Ask every node to forward the packets it has, and wait until every packet is moved by at most $n$ 
  hops (forwarded $n$ times) or gets to the sink. 
  \STATE {\bf Route-Verification:} Every node learns which path it has in the routing graph, i.e., update $\calP(v):=\mbox{$v,r$-path in $\calR_t$}$. 
\ENDFOR
\end{algorithmic}
\end{algorithm}

This entire mechanism can thus be described using two algorithms as follows.
Once activated, a node $v$ updates its routing path $\calP(v)$ using the 
algorithm in Procedure~\ref{algo:activate}. 
The generic protocol is described in Procedure~\ref{algo:protocol}.
This requires an external algorithm $\mathbb{A}$ which 
acts as a {\em scheduler} that generates a permutation -- an order in which nodes will be activated in each round. 
We will assume that these permutations are independent and randomly generated. Our subsequent routing guarantees
will be derived by showing the existence of specific permutations that ensure all packets route.
These permutations are different in each of our models, which differ only in the filtering lists.
We remark that our model is incorporated with a {\em route-verification} step, but this is not a feature of BGP (see the Appendix for a discussion).



With the model defined, we examine the efficiency of packet routing for the three cases of
next-hop preferences with filtering:
\begin{itemize}
\item {\tt General Filtering.} The general case where the filtering
  list $\calD(v)$ of any node $v$ can be an arbitrary subset of nodes.
\item {\tt Not me!} The subcase where the filtering list of node $v$ consists only of itself, $\calD(v)=\{v\}$.
Thus, a node does not want a path through itself, but otherwise has no nodes it wishes to avoid.
\item {\tt Anything Goes!} The case where every filtering list is empty, $\calD(v)=\emptyset$. Thus a node does not 
even mind if its packets cycle back through it!
\end{itemize}

\subsection{Our Results.}

We partition our analyses based upon the types of filtering lists.
Our first result is a strong hardness result presented in Section~\ref{sec:general}.
Not only can it be hard to determine if every packet can be routed but
the maximum number of packets that can be routed cannot be approximated well
even if the network can reach equilibrium. 
Specifically,
\begin{theorem}
\label{thm:general} 
For filtering lists of cardinality at most one, it is NP-hard 
to approximate the maximum cardinality stable subtree to within a
factor of $n^{1-\epsilon}$, for any constant $\epsilon>0$.
\end{theorem}
\begin{corollary}
\label{cor:general} 
For filtering lists of cardinality at most one, given a network in which an equilibrium is guaranteed to exist, it is NP-hard 
to approximate the maximum number of packets that can be routed to 
within a factor of $n^{1-\epsilon}$, for any constant $\epsilon>0$.
\end{corollary}
 
However, for its natural subcase where
the filtering list of a node consists only of itself (that is, a node does not want
to route via a cycle!), we obtain a positive result in Section~\ref{sec:not-me}.
\begin{theorem} 
\label{thm:any}
If the filtering list of a node consists only of itself, then an equilibrium can be obtained in $n$
rounds. Moreover, every packet will be routed in $\frac{n}{3}$ rounds, that is, before stability is obtained!
\end{theorem}

Interestingly, we can route every packet in the case $\calD(v)=\emptyset$ for all $v\in V$; see Section~\ref{sec:anything-goes}. 
Thus, even if nodes do not care whether their packets cycle, the packets still get through!
\begin{theorem} 
\label{thm:not-me}
If the filtering list is empty then every packet can be routed in $4$ rounds, 
even when the network has no equilibrium.
\end{theorem}
Theorems~\ref{thm:any} and~\ref{thm:not-me} are the first theoretical results showing that
packet routing can be done in the absence of stability. 
For example, every packet will be routed even in the presence of {\em dispute wheels}~\cite{GSW02}.
Indeed, packets will be routed even if some nodes {\em never} actually have paths
to the sink.
Note that when we say that every packet will route with probability one we mean that, assuming 
permutations are drawn at random, we will eventually 
get a fair activation sequence that routes every packet. 
It is a nice open problem to obtain high probability guarantees
for fast packet routing under such an assumption.


\section{General Filtering.}\label{sec:general}

Here we consider hardness results for packet routing with general filtering lists.
As discussed, traditionally the theory community has focused upon the stability of $\calR$ -- the routing graph is stable 
if every node is selecting their best valid neighbour (and is consistent). For example, 
there are numerous intractability results regarding whether a network has an equilibrium; e.g., see~\cite{GW99,FP08}.
However, notice that the routing graph may be stable even if it is not spanning!
There may be singleton nodes that prefer to stay disconnected rather than
take any of the offered routes. Thus, regardless of issues such as existence and
convergence, an equilibrium may not even route the packets.
This can be particularly problematic when the nodes use filters. Consider our problem of
maximising the number of nodes that can route packets successfully.
We show that this
cannot be approximated to within a factor of $n^{1-\epsilon}$, for any $\epsilon>0$ unless $\mathrm{P}=\mathrm{NP}$. 
The proof is based solely upon a control plane hardness result: it is NP-hard to approximate the maximum cardinality stable tree to within a factor of $n^{1-\epsilon}$.
Thus, even if equilibria exist, it is hard to determine if there is one in which the {\em sink-component} (the component of $\calR$ containing the sink) is large.

Formally, in the {\em maximum cardinality stable tree problem}, we are given a directed graph $G=(V,E)$ and a sink node $r$; each node $v\in V$ has a ranking of its neighbours and has a filtering list $\calD(v)$. 
Given a tree (arborescence) $T\subseteq G$, we say that a node $v$ is {\em valid} for a node $u$ if $(u,v)\in E$ and a $v,r$-path in $T$ does not contain any node of $\calD(v)$. 
We say that $T$ is {\em stable} if, for every arc $(u,v)$ of $T$, we have that $v$ is valid for $u$, and $u$ prefers $v$ to any of its neighbours in $G$ that are valid for $u$ (w.r.t. $T$).
Our goal is to find a stable tree (sink-component) with the maximum number of nodes.
%
We will show that even when $|\calD(v)|=1$ for all nodes $v\in V$, the maximum-size stable tree problem cannot be approximated to within a factor of $n^{1-\epsilon}$, for any constant $\epsilon>0$,
unless $\mathrm{P}=\mathrm{NP}$.

The proof is based on the hardness of $3$SAT~\cite{Karp72}: given a CNF-formula on $N$ variables and $M$ clauses, it is NP-hard to determine whether there is an assignment satisfying all the clauses.
Take an instance of $3$SAT with $N$ variables,
$x_1,x_2,\ldots,x_N$ and $M$ clauses $C_1,C_2,\ldots,C_M$.
We now create a network $G=(V,A)$ using the following gadgets: 
\begin{itemize}
\item {\tt Variable-Gadget:}
  For each variable $x_i$, we have a gadget $H(x_i)$ with 
  four nodes $a_i,u^T_i,u^F_i,b_i$.
  The nodes $u^T_i$ and $u^F_i$ have first-choice arcs
  $(u^T_i,a_i)$, $(u^F_i,a_i)$ and second-choice arcs
  $(u^T_i,b_i)$, $(u^F_i,b_i)$. 
  The node $a_i$ has two arcs $(a_i,u^T_i)$ and $(a_i,u^F_i)$; the
  ranking of these arcs can be arbitrary.
  Each node in this gadget has itself in the filtering list, i.e.,
  $\calD(v)=\{v\}$ for all nodes $v$ in $H(x_i)$.
\item {\tt Clause-Gadget:}
  For each clause $C_j$ with three variables
  $x_{i(1)},x_{i(2)},x_{i(3)}$, we have a gadget $Q(C_j)$.
  The gadget $Q(C_j)$ has four nodes 
  $s_j,q_{1,j},q_{2,j},q_{3,j},t_j$.
  The nodes $q_{1,j}$, $q_{2,j}$, $q_{3,j}$ have first-choice arcs
  $(q_{1,j},t_j)$, $(q_{2,j},t_j)$, $(q_{3,j},t_j)$. 
  The node $s_j$ has three arcs $(s_j,q_{1,j})$, $(s_j,q_{2,j})$, 
  $(s_j,q_{3,j})$;
  the ranking of these arcs can be arbitrary, so we may assume that
  $(s_j,q_{z,j})$ is a $z$th-choice arc.
  Define the filtering list of $s_j$ and $t_j$ as 
  $\calD(s_j)=\{s_j\}$ and $\calD(t_j)=\{d_0\}$. 
  (The node $d_0$ will be defined later.)
  For $z=1,2,3$, let $u^T_{i(z)}$ and $u^F_{i(z)}$ be nodes in the corresponding Variable-Gadget $H(x_{i(z)})$. 
  The node $q_{z,j}$ has a filtering list $\calD(q_{z,j})=\{u^T_{i(z)}\}$,
  if assigning $x_{i(z)}=\mathsf{False}$ satisfies the clause $C_j$;
  otherwise, $\calD(q_{z,j})=\{u^F_{i(z)}\}$.  
\end{itemize}

To build $G$, we first add a sink node $r$ and a {\em dummy ``sink"} $d_0$. We then connect $d_0$ to $r$ by a first-choice arc $(d_0,r)$. 
We number the Variable-Gadgets and Clause-Gadgets in any order. 
Then we add a first-choice arc from the node $a_1$ of the first Variable-Gadget $H(x_1)$ to the sink $r$. 
For $i=2,3,\ldots,N$, we add a first-choice arc $(b_i,a_{i-1})$ joining gadgets $H(x_{i-1})$ and $H(x_i)$.
We join the last Variable-Gadget $H(x_N)$ and the first Clause-Gadget $Q(C_1)$ by a first-choice arc $(t_1,a_N)$.
For $j=2,3,\ldots,M$, we add a first-choice arc $(t_j,s_{j-1})$ joining gadgets $Q(C_{j-1})$ and $Q(C_j)$.
This forms a line of gadgets. 
Then, for each node $q_{z,j}$ of each Clause-Gadget $Q(C_j)$, we add a second-choice arc $(q_{z,j},d_0)$ joining $q_{z,j}$ to the dummy sink $d_0$.
Finally, we add $L$ {\em padding} nodes $d_1,d_2,\ldots,d_L$ and join each node $d_i$, for $i=1,2,\ldots,L$, to the last 
Clause-Gadget $Q(C_M)$ by a first-choice arc $(d_i,s_M)$; the filtering list of each node $d_i$ is $\calD(d_i)=\{d_0\}$, for all $i=0,1,\ldots,L$.
The parameter $L$ can be any positive integer depending on a given parameter.
Observe that the number of nodes in the graph $G$ is $4N+5M+L+2$, and $|\calD(v)|=1$ for all nodes $v$ of $G$. 
The reduction is illustrated in Figure~\ref{fig:hardness}. 

\begin{figure}[h!]
\begin{center}
  \fbox{\begin{minipage}{\textwidth}
      \centerline{ \includegraphics[scale=0.6] {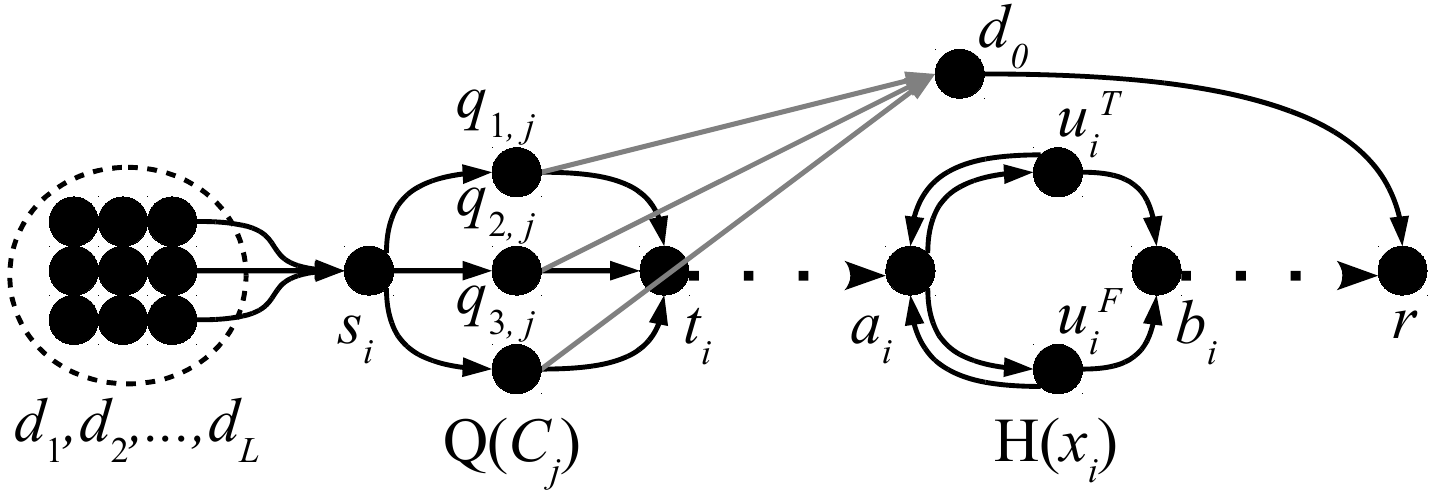}}
      \caption[The hardness construction]{
      The hardness construction. 
      }
      \label{fig:hardness}
    \end{minipage}}
\end{center}
\end{figure}

\medskip

The correctness of the reduction is proven in the next theorem.

\begin{theorem} \label{thm:max-stable}
For any constant $\epsilon>0$, given an instance of the maximum-size stable tree problems with a directed graph $G$ on $n$ nodes 
and filtering lists of cardinality $|\calD(v)|=1$ for all nodes $v$, it is NP-hard to distinguish between the following two cases 
of the maximum cardinality stable tree problem.
\squishlist
\item {\sc Yes-Instance:} The graph $G$ has a stable tree spanning all the nodes.
\item {\sc No-Instance:} The graph $G$ has no stable tree spanning $n^\epsilon$ nodes. 
\squishend
\end{theorem}

\begin{proof}
We apply the above reduction from $3$SAT with a parameter $L=J^{1/\epsilon}-J$, where $J=4n+5m+2$.
Thus, the graph $G$ has $n=J^{1/\epsilon}$ nodes and has $n^\epsilon=J$ non-padding nodes.

First, we show that there is a one-to-one mapping between choices of each Variable-Gadget $H(x_i)$ and an assignment of $x_i$.
Consider any Variable-Gadget $H(x_i)$.
To connect to the next gadget, nodes $u^T_i$ and $u^F_i$ of $H(x_i)$ must choose at least one second-choice arc.
However, in a stable tree, they cannot choose both second-choice arcs
$(u^T_i,b_i)$ and $(u^F_i,b_i)$; otherwise,
$u^T_i$ or $u^F_i$ would prefer to choose the node $a_i$.  
Thus, the gadget $G(x_i)$ must choose either arcs
\begin{center}
(1)~$(u^T_i,b_i), (u^F_i,a_i), (a_i,u^T_i)$ \quad or \quad 
(2)~$(u^F_i,b_i), (u^T_i,a_i), (a_i,u^F_i)$.
\end{center}
These two cases correspond to the the assignments
$x_i=\mathsf{True}$ and $x_i=\mathsf{False}$, respectively.
Thus, there is a one-to-one mapping between the choices of gadget $H(x_i)$
in the stable tree and the assignment of $x_i$. 
We refer to each of these two alternatives as an assignment of $x_i$. 

Now, we prove the correctness of the reduction.

\medskip

\noindent {\sc Yes-Instance:}
Suppose there is an assignment satisfying all the clauses. 
Then there is a stable tree $T$ corresponding to such an assignment. To see this, within Variable-Gadget
we select arcs in accordance with the assignment as detailed above.
We also choose the arc $(d_0,r)$ and all the horizontal arcs connecting adjacent gadgets in the line (or from the first Variable-Gadget to the sink $r$).
For each Clause-Gadget $Q(C_j)$ and each $z=1,2,3$, we choose the first-choice arc $(q_{z,j}, t_j)$ if the assignment to $x_{i(z)}$ 
satisfies $C_j$; otherwise, we choose the second-choice arc $(q_{z,j}, d_0)$.
For the node $s_j$ of $Q(C_j)$, we choose an arc $(s_j,q_{z,j})$, where $z$ is the smallest number such that the assignment to $x_{i(z)}$ 
satisfies $C_j$ (i.e., $q_{z,j}$ chooses $t_j$); since the given assignment satisfies all the clauses, $s_j$ has at least one valid choice. 
Now, we have that the node $s_M$ of the last Clause-Gadget $Q(C_j)$ has a path $\calP(s_M)$ to the sink $r$ that does not contain the dummy sink $d_0$. 
Thus, every padding node can choose $s_M$ and, therefore, is in the stable tree $T$. This implies that $T$ spans all the nodes.

\medskip

\noindent {\sc No-Instance:}
Suppose there is no assignment satisfying all the clauses. 
Let $T$ be any stable tree of $G$. 
As in the previous discussion, the choices of nodes in Variable-Gadgets correspond to the assignment of variables of $3$SAT.

Consider any Clause-Gadget $Q(C_j)$. 
Since $\calD(s_j)=\{d_0\}$, the node $s_j$ of $Q(C_j)$ has a path to the sink $r$ only if
\begin{enumerate}[(1)]
\item a $t_j,r$-path $\calP(t_j)$ in $T$ does not contain the dummy sink $d_0$, and 
\item one of $q_{1,j},q_{2,j},q_{3,j}$ chooses $t_j$.
\end{enumerate}

These two conditions hold only if $T$ corresponds to an assignment satisfying $C_j$.
To see this, suppose the first condition holds.
Then $\calP(t_j)$ has to visit either $v^F_i$ or $v^T_i$ of every Variable-Gadget $H(x_i)$, depending on the assignment of $x_i$. 
Thus, by the construction of $\calD(q_{z,j})$, $t_j$ is valid for $q_{z,j}$ only if the assignment to $x_{i(z)}$ satisfies $C_j$.
Since there is no assignment satisfying all the clauses, a node $s_{\ell}$ of some Clause-Gadget $Q(C_\ell)$ is not in $T$.
This means that nodes in the remaining Clause-Gadget have to use the dummy sink $d_0$ to connect to the sink $r$. 
Thus, the node $s_M$ of the last Clause-Gadget $Q(C_M)$ is not in $T$ and neither are any of the padding nodes $d_1,d_2,\ldots,d_L$. 
Therefore, the size of $T$ is at most $J=n^\epsilon$, proving the theorem.
\end{proof}

Observe this means that, from the perspective of the nodes, it is NP-hard to determine whether adding an extra node to its 
filtering list can lead to solutions where none of its packets ever route. 
In other words, it {\bf cannot} avoid using an intermediate node it dislikes!




\section{Filtering: Anything-Goes!}
\label{sec:anything-goes}

Here we consider the case where every node has an empty filtering list.
This case is conceptually simple but still contains many 
technical difficulties involved in tracking packets when nodes become mistaken in their connectivity assessments.
In this case, networks with no equilibrium can exist. Figure~\ref{fig:no-stable} presents such an example.
Moreover, in this example, fair activation sequences exist where the node $v$ will
never be in the sink component; for example, repeatedly activate nodes
according to the permutation $\{v,u,w,x,y\}$.
Despite this, every packet will route in two rounds!
This example nicely illustrates the need to track packets if we want to understand
the efficacy of BGP-like protocols.
\begin{figure}[h!]
  \fbox{
  \begin{minipage}{\textwidth}
      \centerline{ \includegraphics[scale=0.6] {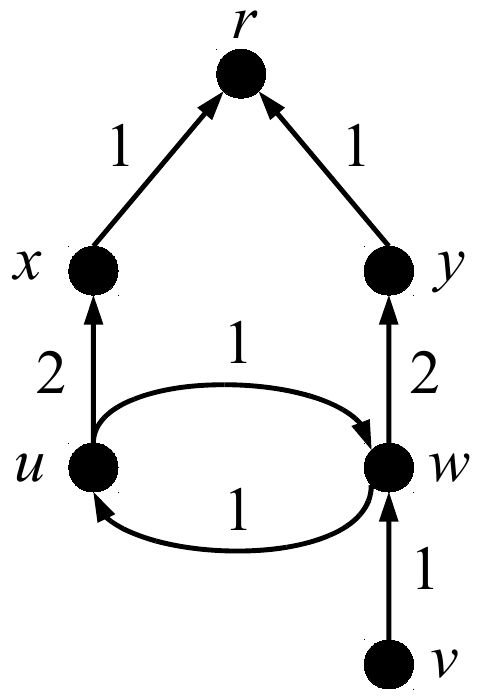} }
      \caption[A network with no stable spanning tree.]  {A network with no stable spanning tree.
        (Arc numbers indicate rankings, e.g., the number
        $2$ on the arc $(u,x)$ means that $x$ is the second choice of
        $u$.)  }
      \label{fig:no-stable}
    \end{minipage}
  }
\end{figure}

For this class of preference functions, we show that the property of successful routing is universal.
In any network, every packet will reach the sink. Specifically, we present a fair activation sequence 
of four rounds that routes every packet, even when there is no equilibrium.

Observe that when filtering lists are empty, a node $v$ only needs to known whether its neighbour $u$ has a path to the sink,
as $v$ will never discount a path because it contains a node it dislikes.
Thus, we can view each node as having two states: clear or opaque.
A node is {\em clear} if it is in the sink-component (the nomenclature derives from the fact that a packet at such a node will then reach the sink -- that is, ``clear"); otherwise, a node is {\em opaque}. 
Of course, as nodes update their chosen next-hop over time, they may be mistaken in their beliefs (inconsistent) as the routing graph changes.
In other words, some clear nodes may not have ``real'' paths to the sink.  
After the learning step at the end of the round, these clear-opaque states become correct again.
Our algorithm and analysis are based on properties of the
network formed by the first-choice arcs, called the {\em first class network}. 
We say that an arc $(u,v)$ of $G$ is a {\em first-choice} arc if $v$ is the most preferred neighbour of $u$. 
We denote the first class network by $F=(V,A_1)$, where $A_1$ is the
set of the first-choice arcs.
As in a routing graph $\calR$, every node in $F$ has one outgoing arc. 
Thus, every component of $F$ is a {\em $1$-arborescence}, a tree-like structure with either a cycle or a single node as a root.
We denote the components of $F$ by $F_0,F_1,\dots, F_\ell$, where $F_0$ is the component containing the sink $r$.
Each $F_j$ has a unique cycle $C_j$, called a {\em first class cycle}. 
We may assume the first class cycle in $F_0$ is a self-loop at the sink $r$.
Furthermore, when activated, every node in $F_0$ will always choose its neighbour in $F_0$; so, we may assume 
wlog that $F_0$ is the singleton node $\{r\}$.
The routing graph at the beginning of round $t$ is denoted by $\calR_t$. 
We denote by $\calK_t$ and  $\calO_t$ the set of clear and the
set of opaque nodes at the start of round $t$. 
%
Now we show that there is an activation sequence which routes every packet in four rounds.

The proof has two parts: a {\em coordination phase} and a {\em routing phase}.
In the first phase, we give a coordination algorithm that generates 
a permutation that gives a red-blue colouring of the nodes 
with the following two properties:
\begin{enumerate}[(i)]
\item For each $F_j$, every node in $F_j$ has the same colour, i.e.,
the colouring is coordinated.
\item If the first class cycle $C_j$ of $F_j$ contains a clear node, then all nodes in $F_j$ must be coloured blue.
\end{enumerate}
We remark that, subject to the these two properties, our algorithm will maximise the number of nodes coloured red, but this is not 
needed to prove that we can route a packet successfully.

The usefulness of this colouring mechanism lies in the fact that the corresponding permutation 
is a fair activation sequence that will force the red nodes to lie in the sink-component and the
blue nodes to lie in non-sink components. Moreover, bizarrely, running this coordination algorithm four 
times in a row ensures that every packet routes!
So, in the second phase (the routing phase), we simply run the coordination algorithm three more times.
\subsection{The Coordination Phase.}
The algorithm Coordinate$(\calK_t)$ presented in Procedure~\ref{algo:coord} constructs a red-blue colouring of the nodes,
i.e., a partition $(R,B)$ of $V$ (where $v\in R$ means that $v$ is coloured red and $v\in B$ means that $v$ is coloured blue) and 
which has the property that 
any node $v\in R$ prefers some node in $R$ to any node $w\in B$, and
any node $v\in B$ prefers some node in $B$ to any node $w\in R$.

\begin{algorithm}
\caption{Coordinate($\calK_t$)}
\label{algo:coord}
\begin{algorithmic}[1]
\REQUIRE A set of clear nodes $\calK_t$. 
\ENSURE A partition $(R,B)$ of $V$.
\STATE Let $B_0:=\bigcup_{i: V(C_j)\cap\calK_t\neq\emptyset}V(F_i)$ be a set of nodes contained in an $F$-component whose first class cycle $C_i$ has a clear node.
\STATE Initialise $q := 0$. 
\REPEAT
   \STATE Update $q := q + 1$.
   \STATE Initialise $B_q := B_{q-1}$, $R_q := \{r\}$ and
   $U:=V\setminus(R_q\cup B_q)= R_{q-1}\setminus\{r\}$.
   \WHILE{$\exists$ a node $v\in U$ that prefers a node in $R_q$ to nodes in $B_{q-1}\cup(U\cap \mathcal{K}_t)$}
      \STATE Move $v$ from $U$ to $R_q$.
   \ENDWHILE
   \STATE Add all the remaining nodes in $U$ to $B_q$. \label{step:add-blue-nodes}
\UNTIL{$B_q=B_{q-1}$.}
\RETURN $(R_q,B_q)$. 
\end{algorithmic}
\end{algorithm}

Observe that Coordinate$(\calK_t)$ contains many loops. However, we only wish to generate a 
single activation sequence $\pi_t$ from the procedure. To do this, we take the output partition $(R,B)$
and use it to build an activation sequence.

We begin by activating nodes in $B_0=\bigcup_{i:C_i\cap \mathcal{K}_i\neq\emptyset}V(F_i)$, 
the components $F_i$ whose first class cycles contain at least one clear node. 
To order the nodes of $B_0$, we proceed as follows. For each $F_i$ with $V(F_i)\subseteq B_0$, take a 
clear node $v\in C_i\cap\calK_t$. Then activate the nodes of $F_i$ (except $v$) in increasing order of 
distance from $v$ in $F_i$, and after that activate $v$. 
This forms a non-sink-component $F_i$ in the routing graph as every node can choose its first-choice. 
Next, we activate the remaining nodes in $B$.
We order the nodes of $B\setminus B_0$ in a greedy fashion; a node can be activated once 
its most-preferred {\em clear} neighbour is in $B$. 
Finally, we activate the nodes in $R$. Again, this can be done greedily. 
Specifically, we activate nodes of $R$ in the same order as when they were added to $R$.

%
%

\begin{lemma}
\label{lmm:form-components}
Given a partition $(R,B)$ from Coordinate($\calK_t$), 
the activation sequence $\pi_t$ induces a sink-component on $R$ and
non-sink-components on $B$.  
\end{lemma}

\begin{proof}
First, let us verify by induction that each node in $R$ ends up in the sink-component.
For the base case, the sink $r$ is clearly in the sink-component.
Now, the red nodes are activated after the blue nodes. 
Let $q^*$ be the last iteration of Coordinate$(\calK_t)$; hence, $B_{q^*}=B_{q^*-1}$. 
So, at the start of this iteration, $B_{q^*}=B$ and $R_{q^*}=\{r\}$. 
Despite this, at the end of the round, we have $R_{q^*}=R$.
This implies that every node $v$ of $R$ prefers a node of $R_{q^*}$ to a node of 
$B_{q^*-1}\cup ( U\cap \mathcal{K}_t) = B\cup ( U\cap \mathcal{K}_t) \supseteq B$
when it is added to $R_{q^*}$.
But, then $\pi_t$ orders the nodes such that, on activation, every node $v$ of $R$ prefers
a node of $R_{q^*}$ (which are in the sink component by induction) to a node of $B$.
Thus, regardless of which components the blue nodes were placed in upon activation,
all the red nodes are placed in the sink-component
(it can only help if the some of the red nodes were placed in the sink-component).
%

Now, let us show that the blue nodes do not end up in the sink-component.
By the construction of $\pi_t$, the nodes of $B_0$ connect together via their first class components.
Consequently, they belong to non-sink components.
Finally, consider the nodes in $B\setminus B_0$. 
Observe that the size of $B_q$ increases in each iteration, that is, $B_q\subsetneq B_{q-1}$ for all $q< q^*$. 
So, at the end of their iterations, $R_q\subseteq R_{q-1}$. 
Now, take a node $v$ added to $B_q$ in Step~\ref{step:add-blue-nodes} in iteration $q$. 
At this point, the remaining nodes are $U=R_{q-1}\setminus R_{q}$, 
and $v$ prefers some node $w$ in $B_{q-1}\cup(U\cap \mathcal{K}_t)$ to any node 
in $R_q$ (or has no neighbours in $R_q$ at all).
But, $B_{q-1}\cup(U\cap \mathcal{K}_t) = B_{q-1}\cup(R_{q-1}\setminus R_{q} \cap \mathcal{K}_t)
\subseteq B_{q-1}\cup(R_{q-1}\setminus R_{q}) =B_q$.
Thus, $v$ either prefers some node in $B_{q^*}$ to any node in $R_{q^*}$ or
has no neighbours in $R_{q^*}$.
But, by assumption, every node has a path to the sink in the ``all-choice'' graph $G$. 
So, when considering the blue nodes, there must be a node in $B\setminus B_0$ 
that has its most preferred ``clear'' neighbour in $B$.
Therefore, $\pi_t$ induces non-sink components on $B$.
\end{proof}

The coordination and other desirable properties hold when we apply $\pi_t$.
\begin{lemma}
\label{lmm:act-coord}
Given the activation sequence $\pi_t$ for $(R,B)$, at the end of the round, the following hold:
\squishlist
\item The sink-component includes $R$ and excludes $B$.
\item {\tt Coordination:} For each $F_i$, either all the nodes of $F_i$ are in the sink-component or none of them are.
\item Let $B_0=\bigcup_{i:V(C_i)\cap\calK_t\neq\emptyset}V(F_i)$ and suppose $\calK_t=B_0$. If a packet travels for $n$ hops 
but does not reach the sink, then it must be at a node in $\calK_t$.
\squishend
\end{lemma}
\begin{proof}
The first statement follows Lemma~\ref{lmm:form-components}.

For the second statement, it suffices to show that, for each $F_i$, either $V(F_i)\subseteq R$ or $V(F_i)\subseteq B$.
If not, then there is an $F_i$ containing both red and blue nodes.
Thus, there are two possibilities.\\
(i) There is a node $v\in F_i$ such that $v$ is in $R$ and its first choice is $w\in B$.
This is not possible. To see this, observe that we must have $w\in B_{q-1}$ 
because of the stopping condition $B_q=B_{q-1}$.
But, then $v$ could not be added to $R_q$, a contradiction.
(ii) There is a node $v\in F_i$ such that $v$ is in $B$ and its first choice is $w\in R$.
Again, this is not possible. To see this, observe that we must have 
$v\notin B_0$ because each $B_0$ consists only of first class components that are 
monochromatic blue. Recall also that $R_q\subseteq R_{q-1}$ for all $q$. 
Thus, we must have $w\in R_q$ for all $q$. But, then $v$ would have been added to $R_q$,
a contradiction.
%
For the third statement,
note that a packet that travels for $n$ hops but does not reach the sink must be stuck in some cycle. 
Consider the construction of $(R,B)$. Since $\calK_t=B_0$, we only add a node to $B$ whenever it prefers some node in $B$ to any node in $R$. 
Because $U\cap\calK_t=\emptyset$, nodes in $B\setminus B_0$ cannot form a cycle on their own.
Thus, the packet is stuck in a cycle that contains a clear node; the only such cycles are the first class cycles of $B_0$
since $\calK_t=B_0$.
\end{proof}

The following lemma follows by the construction of a partition $(R,B)$.

\begin{lemma} \label{lmm:prop-coord}
Let $(R',B')$ be any partition generated from the procedure Coordinate(\,), and 
let $(R,B)$ be a partition obtained by calling Coordinate($\calK_t$) where $\bigcup_{i:V(C_i)\cap\calK_t\neq\emptyset}V(F_i)\subseteq B'$.
Then $R'\subseteq R$.
\end{lemma}
\begin{proof}
Consider a partition $(R_q,B_q)$ constructed during a call to Coordinate($\calK_t$). 
Observe that $B_0\subseteq B'$ because 
$B_0=\bigcup_{i:V(C_i)\cap\calK_t\neq\emptyset}V(F_i)$.
By the construction of $(R',B')$, since $B_0\subseteq B'$, every
node of $R'$ must have been added to $R_1$, i.e., $R'\subseteq R_1$. 
Inductively, if $R'\subseteq R_q$ for some $q\geq 1$, then
$B_q\subseteq B'$ and thus $R'\subseteq R_{q+1}$ by the same
argument. 
\end{proof}

\subsection{The Routing Phase }

%
Running the coordination algorithm four times ensures every packet will have 
been in the sink-component at least once, and thus, every packet routes.
%

\begin{theorem}\label{thm:route-four-rounds}
 In four rounds, every packet routes.
\end{theorem}
\begin{proof}
The first round $t=1$ is simply the coordination phase.
We will use subscripts on $R$ and $B$ (e.g., $R_t$ and $B_t$) to denote the final colourings output in each round 
and not the intermediate sets $R_q$ and $B_q$ used in Coordinate(\,).
Now, consider a packet generated by any node of $V$. 
First, we run Coordinate$(\calK_1)$ and obtain a partition $(R_1,B_1)$. 
By Lemma~\ref{lmm:act-coord}, if the packet is in $R_1$, then it is routed successfully, and we are done.
Hence, we may assume that the packet does not reach the sink and 
thus must be in $B_1$. 
Note that, now, each $F_i$ is either contained in $R_1$ or $B_1$ by Lemma~\ref{lmm:act-coord}.

We now run Coordinate$(\calK_2)$ and obtain a partition $(R_2,B_2)$.
By Lemma~\ref{lmm:act-coord}, $\calK_2=R_1$. 
So, if the packet does not reach the sink, it must be in $B_2$.
Since no first class component intersects both $R_1$ and $B_1$, we have 
$R_1=\calK_2=\bigcup_{i:V(C_i)\cap\calK_2\neq\emptyset}V(F_i)$.
Thus the nodes $\calK_2$ are all initially coloured blue in Step 1 of Coordinate$(\calK_2)$.
As the set of blue nodes only expands throughout the round

So, $R_1\subseteq B_2$ (since $\calK_2\subseteq B_2$) and $R_2\subseteq B_1$, and Lemma~\ref{lmm:act-coord} implies 
that the packet is in $R_1$.

Third, we run Coordinate$(\calK_3)$ and obtain a partition $(R_3,B_3)$.
Applying the same argument as before, 
we have that the packet is in $R_2$ (or it is routed), $R_2\subseteq B_3$
and $R_3\subseteq B_2$. 
Now, we run Coordinate$(\calK_4)$ and obtain a partition $(R_4,B_4)$.
By Lemma~\ref{lmm:act-coord}, we have 
$\calK_4=\bigcup_{i:V(C_i)\cap\calK_4\neq\emptyset}V(F_i)$. 
Since $R_3=\calK_4\subseteq B_2$, Lemma~\ref{lmm:prop-coord} implies that $R_2\subseteq R_4$.
Thus, the packet is routed successfully since $R_4$ is contained in the sink-component. 
\end{proof}



\section{Filtering: Not-Me!}
\label{sec:not-me}

In practice, it is important to try to prevent cycles forming in the routing graph of a network. To achieve this, {\em loop-detection} 
is implemented in the BGP-4 protocol~\cite{Ste98}.
The ``Not-Me!'' filtering encodes loop-detection in the BGP-4 protocol simply by having a filtering list $\calD(v)=\{v\}$, for every node $v$.
For this class of preference function, we again show that every packet will route. Recall, this is in 
contrast to Theorem \ref{thm:max-stable}, which says that it is NP-hard to determine whether we can 
route every packet for general filtering lists of cardinality one.
Moreover, we exhibit a constructive way to obtain a stable spanning
tree via fair activation sequences.
Interestingly, all of the packets will have routed before stability is obtained.
In particular, we give an algorithm that constructs an activation sequence such that every packet routes successfully in $\frac13 n$ rounds
whereas the network itself becomes stable in $n$ rounds.

This result is the most complicated part of our paper, so we will first give a high level overview.
Clearly, when filtering lists are non-empty, we have an additional difficulty: even if $w$ is the most preferred choice 
of $v$ and $w$ has a non-empty routing path $P$, $v$ still may not be able to choose $w$ because $P$ contains
a node on $v$'s filter list (in this case, $v$ itself). This can cause the routing graph to evolve in ways that are very difficult to
keep track of. Thus, the key idea is to design activation permutations that manipulate the routing graph
in a precise and minor fashion in each round. To do this, we search for a spanning tree with a {\em Strong Stability Property}. 
\smallskip


\begin{property}[Strong Stability - Informal]
A spanning tree $S$ has the strong stability property on
$\mathbb{O}\subseteq V$ if and only if, for every node $v\in\mathbb{O}$,  
the most preferred choice of $v$ is its parent $w$ in $S$,
even if $v$ can choose any node outside $\mathbb{O}$ {\bf and} 
any node outside its subtree in the forest $S[\mathbb{O}]$.
 \end{property}


To illustrate this property, consider a simple setting where $S$ is just the path $S=(e,d,c,b,a,r)$ and $\mathbb{O}=\{b,c,e\}$. 
If $S$ is strongly stable on $\mathbb{O}$, then $b$ must prefer $a$ to nodes in $\{r,d,e\}$.
Observe that $e$, whilst  a descendent of $b$ in $S$, is {\em not} a descendent of $b$ in $S[\{b,c,e\}]$. 
So, even if $b$ is allowed to choose $e$, which is a descendent of $b$ in $S$, $b$ still wants to choose $a$ as its parent. 

Thus, the strong stability property says that the choice of a vertex
$v\in\mathbb{O}$ in $S$ is the best one even if 
all the nodes outsides $\mathbb{O}$ change their choices. 
(that is, even if we replace $S[V\setminus\mathbb{O}]$ with a
completely different subgraph). 
%
For the special case where $\mathbb{O}= \calO$, the set of opaque nodes, 
if we activate nodes of $S$ in increasing order of distance from the sink $r$ then
every node in $\calO$ will choose its parent in $S$ -- as the clear nodes 
in $S_v$ are not desirable to connect to.
As we will see, under certain conditions, we can even maintain the choices of nodes
in $\mathbb{O}$ even if some of them are clear and some are opaque.

\smallskip

A {\em stable spanning tree}, a tree where no node wants to change its choice, 
can be found in polynomial-time, 
and given a stable spanning tree $S$, it is easy to force opaque nodes in $\calO_t$ to make the same choices as in $S$.
But, this only applies to the set of opaque nodes, which changes with each round.
The strong stability property allows us to make a stronger manipulation. 
Intuitively, the strong stability property says that once we force every node $v\in\mathbb{O}$ to make the same choice as in $S$, 
we can maintain these choices in all the later rounds.
Moreover, in each round, if we cannot route all the packets, then we can make the strong stability property spans 
three more nodes; otherwise, the property spans one more node.
Thus, in $\frac13 n$ rounds, every packet will route, but we need $n$ rounds to obtain stability.

With this overview complete we introduce some formal definitions needed for the proof.
Again, $\calO_t$ and $\calK_t$ denote the set of opaque and clear nodes, respectively, at the beginning 
of round $t$.
Given a graph $R$ and a set of nodes $U$, we denote by
$R[U]=\{(u,v):u,v\in U, (u,v)\in R\}$ the subgraph of $R$ induced by
$U$, and we denote by $\mathcal{A}_{R}^+(U) =\{(u,v):u\in U, (u,v)\in R\}$, the subgraph of arcs of $R$ induced by $U$ plus arcs leaving $U$.
%

Given a set of nodes $Q\subseteq{V}$, the {\em $Q$-subtree} of $v$, with respect to a tree $T$, is the maximal subtree rooted at $v$ 
of the forest $T[Q]$.
%
A (spanning) tree $T$ is {\em stable} if every node $v$ with 
$(v,w)\in T$, prefers its parent $w$ in $T$ to every non-descendant; 
thus, no node wants to change its next hop.
%
A spanning tree $S$ has the {\em strong stability property} on the set
of nodes $\mathbb{O}$ if every node $v\in\mathbb{O}$ with $(v,w)\in S$
prefers its parent $w$ to every node outside its
$\mathbb{O}$-subtree; observe that if $\mathbb{O}=V$ (or
$V\setminus\{r\}$), then $S$ is also stable.
%
We say that $S$ is a {\em skeleton} of a (non-spanning) tree $T$ if,
for every (maximal) subtree $F\subseteq S[\mathbb{O}]$, either
$T$ contains $\mathcal{A}^+_{S}(F)$ or $T$ contains no node of $F$.
%



\subsection{Finding a Strongly Stable Tree}
\label{sec:find-strong-tree}

In this section, we present a subroutine for finding a spanning tree
with the strong stability property.
The input of this algorithm (see Procedure~\ref{algo:findstable}) is a sink-component $T^{in}$ and a
spanning tree $S^{in}$ with the strong stability property on a given
set of nodes $\mathbb{O}$. 
The algorithm expands the strong stability property to also hold on $\calO$, 
the set of nodes not in the sink-component $T^{in}$. 

\begin{algorithm}
\caption{FindStable($T^{in}$, $S^{in}$, $\mathbb{O}$)}
\label{algo:findstable}
\begin{algorithmic}
\REQUIRE A sink-component $T^{in}$ and a spanning (or empty)
         tree $S^{in}$ such that\\
         (1)~The tree $S^{in}$ has the strong stability property on
         $\mathbb{O}$, and\\
         (2)~$S^{in}$ is a {\bf skeleton} of $T^{in}$. 
\ENSURE A stable spanning tree $S^{out}$ with the strong stability
        property on $\mathbb{O}\cup\calO$, where
        $\calO=V\setminus V(T^{in})$.
\end{algorithmic}
\begin{algorithmic}[1]
\STATE Let $\calO=V\setminus V(T^{in})$ be the set of nodes not in the
       sink component $T^{in}$.
\STATE Initialise $S^{out}:=T^{in}\cup\mathcal{A}^+_{S^{in}}(\calO)$.
\STATE Initialise $\mathcal{C}_1:=S^{out}[\calO]$.
\FOR{iteration $t:=1$ to $|\calO|$}
  \STATE Pick an arbitrary leaf $v$ of $\mathcal{C}_t$.
  \STATE Pick a node $w\in V(S^{out})$ such that $v$ prefers $w$ to
         any other node not in its $\calO$-subtree in $S^{out}$.  
         \label{step:change-nexthop}
  \STATE Replace the arc $(v,y)$ in $S^{out}$ by the arc $(v,w)$. 
  \STATE Update $\mathcal{C}_{t+1} := \mathcal{C}_t\setminus\{v\}$.
\ENDFOR
\RETURN $S^{out}$. 
\end{algorithmic}
\end{algorithm}


Before proving the correctness of the procedure 
FindStable($T^{in}$,$S^{in}$, $\mathbb{O}$), we prove some basic
facts.

\begin{lemma}[Union Lemma]
\label{lmm:union-of-stability}
Let $S$ be a spanning tree that is strongly stable on
$A\subseteq V$ and also on $B\subseteq V$.
Then $S$ is strongly stable on $Q=A\cup{B}$.
\end{lemma}

\begin{proof}
Without loss of generality, take a vertex $v\in A \subseteq Q$.
Let $F_A$ and $F_Q$ be the (maximal) $A$-subtree and $Q$-subtree of
$v$ in $S$, respectively. 
Then $V(F_A)\subseteq V(F_Q)$ because $A\subseteq Q$.
By the strong stability property of $S$ on $A$, we have that 
$v$ prefers its parent $w$ in $S$ to any other node in 
$V\setminus V(F_A)$.
But, $V\setminus V(F_Q) \subseteq V\setminus V(F_A)$. It follows that
$S$ is strongly stable on $Q=A\cup{B}$. 
\end{proof}

The next lemma proves an important property of a skeleton of a tree $T$. 
\begin{lemma}[Skeleton Lemma]
\label{lmm:skeleton}
Let $T$ be any tree. Let
$S$ be a spanning tree that is strongly stable on a set of
vertices $\mathbb{O}$ and be a skeleton of $T$.
Let $\calO=V\setminus V(T)$.
Then, for any spanning tree $T'$ such that $T\subseteq T'$, 
the tree $T'$ is strongly stable on $\mathbb{O}\setminus\calO$.
\end{lemma}

\begin{proof}
Consider any node $v\in\mathbb{O}\setminus\calO$.
Let $F_v$ be the $\mathbb{O}$-subtree of $v$ in $S$. 
By the definition of skeleton, 
for any (maximal) subtree $F\subseteq S[\mathbb{O}]$, either
(i)~$\mathcal{A}^+_{S}(F)\subseteq T$ or 
(ii)~$V(F)\cap V(T)=\emptyset$.
Since $v\in V(T)\cap  \mathbb{O}$, it must be that $\mathcal{A}^+_{S}(F_v) \subseteq T\subseteq T'$.
Thus, $V(F_v)\subseteq \mathbb{O}\setminus\calO$ as
$V(T)=V-\calO$. 
Therefore $F_v$ is also the $(\mathbb{O}\setminus\calO)$-subtree of $v$ in $S$.
By the strong stability property of $S$ on $\mathbb{O}$, we know that
$v$ prefers its parent $w$ in $S$ to any node in $V\setminus V(F_v)$. 
So, $S$ has the strong stability property on
$\mathbb{O}\setminus\calO$.
\end{proof}

The next lemma shows the correctness of the procedure
Stabilise$(T^{in},S^{in},\mathbb{O})$. 

\begin{lemma} \label{lem:ext-stability}
The procedure FindStable$(T^{in},S^{in},\mathbb{O})$ outputs a spanning tree $S^{out}$ with the strong stability property on $\mathbb{O}\cup\calO$.
\end{lemma}

\begin{proof}



To begin, we show that $S^{out}$ is a spanning tree throughout the
procedure. Initially $S^{out}=T^{in}\cup\mathcal{A}^+_{S^{in}}(\calO)$. 
Therefore, $S^{out}$ contains every node, since $\calO=V-V(T^{in})$.
Let us see that $S^{out}$ is also connected.
As $S^{in}$ is a spanning tree, we have that each component $F$ of $\mathcal{A}^{+}_{S^{in}}(\calO)$ is a 
tree (arborescence). Furthermore as the sink $r$ is not in $\calO$, there is exactly one
arc in $\mathcal{A}^{+}_{S^{in}}(\calO)$ leaving $F$ (from its root) and entering $V(T^{in})$.
Thus, $S^{out}=T^{in}\cup\mathcal{A}^+_{S^{in}}(\calO)$ is a spanning tree.


Now, consider how $S^{out}$ changes during the loop phase of the procedure. 
No node in $V-\calO$ is considered during this phase, so $S^{out}$ never contains 
any arc leaving $V\setminus\calO$ and entering $\calO$.
As a result, the $\calO$-subtree of $v$ in $S^{out}$ coincides exactly with the
set of all descendants of $v$ in $S^{out}$.
Hence, in Step~\ref{step:change-nexthop}, node $v$ never selects a descendant node to be $w$. 
So, we can safely replace the arc $(v,y)\in S^{out}$ by the arc $(v,w)$ without
creating a cycle. This shows that $S^{out}$ is always a spanning tree.

Next, we show that $S^{out}$ has the strong stability property 
both on $\mathbb{O}\setminus\calO$ and on $\calO$. By Lemma~\ref{lmm:union-of-stability}, this will
imply that $S^{out}$ is strongly stable on $\mathbb{O}\cup\calO$.
Now, $S^{in}$ is strongly stable on $\mathbb{O}$ and is a skeleton on $T^{in}$.
Furthermore, $T^{in}\subseteq S^{out}$ by construction. Thus, applying Lemma~\ref{lmm:skeleton},
we have that $S^{out}$ is strongly stable on $\mathbb{O}\setminus\calO$. 
It only remains to show that $S^{out}$ is strongly stable on $\calO$.
To achieve this, we show by induction that $S^{out}$ is strongly stable on
$\calL_{t}=\calO\setminus V(C_t)$ in each iteration $t$. (Note that, on termination, $\calL_t=\calO$.)
This is true for $t=1$ as $\calL_1=\emptyset$.
Now, consider iteration $t>1$, and assume that strong stability holds on $\calL_{t-1}$. 
Observe that no node $u\in\calO\setminus\calL_{t-1}$ has a parent
$x\in\calL_{t-1}$; otherwise, $x$ would have not been added to
$\calL_{t-1}$. 
Since $v$ is a leaf of $C_t$, all the nodes in the $\calO$-subtree of $v$ 
in $S^{out}$ must be in $\calL_{t-1}$.
Because nodes in $\calL_{t-1}$ can not change their parents after this time,
every descendant of $v$ in $\calO$ will remain a descendant of $v$.
Consequently, $v$ prefers $w$ to other any non-descendant in $S^{out}$
throughout the rest of the procedure.
Thus, $S^{out}$ is strongly stable on $\calL_t$.
\end{proof}

\subsection{Routing Every Packet in $n$ Rounds.}

We are now ready to present an algorithm that routes every
packet in $n$ rounds (recall that each round consists of a 
single fair-activation sequence).
In addition to the procedure FindStable(), two procedures (namely, Procedures~\ref{algo:bfs} and \ref{algo:rev-bfs}) 
based upon a breath-first-search (BFS)
algorithm are our basic building block for generating an activation sequence. 
Given a spanning tree $F$ and a set of nodes $U\subseteq V$, 
the procedure BFS$(U,F)$ activates the
nodes of $U$ in {\em breadth-first-search (BFS) order}. That is, the nodes of $U$ are activated in increasing
order of distance to the sink $r$ in $F$.

\begin{algorithm}
\caption{BFS$(U,F)$}
\label{algo:bfs}
\begin{algorithmic}[1]
\REQUIRE A set $U\subseteq V$ and a spanning tree $F$.
\STATE Let $v_1,v_2,\ldots,v_q$ be nodes in $U\setminus\{r\}$
       sorted in increasing order of distance to the root $r$ of $F$. 
\FOR{$i := 1$ to $q$}
  \STATE Activate $v_i$.
\ENDFOR
\end{algorithmic}
\end{algorithm}

Similarly, the procedure reverse-BFS$(U,F)$ activates the
nodes of $U$ in breadth-first-search (BFS) reverse-order.
\begin{algorithm}
\caption{reverse-BFS$(U,F)$}
\label{algo:rev-bfs}
\begin{algorithmic}[1]
\REQUIRE A set $U\subseteq V$ and a spanning tree $F$.
\STATE Let $v_1,v_2,\ldots,v_q$ be nodes in $U\setminus\{r\}$
       sorted in increasing order of distance to the root $r$ of $F$. 
\FOR{$i := q$ to $1$}
  \STATE Activate $v_i$.
\ENDFOR
\end{algorithmic}
\end{algorithm}

Over the course of these $n$ rounds, the main algorithm (Procedure~\ref{algo:fair-stabilize})  utilises these 
three procedures on the following two classes of nodes.
\begin{enumerate}[(1)]
\item The set of nodes that have been clear in {\em every} round up to
  time~$t$, denoted by $\mathbb{K}_t =\cap_{i=1}^t \calK_i$.
\item The complement of $\mathbb{K}_t$, which is the set of nodes that
  have been opaque at least once by time~$t$, denoted by $\mathbb{O}_t =\cup_{i=1}^t\calO_i$.
\end{enumerate}

\begin{algorithm}
\caption{Fair-Stabilise()}
\label{algo:fair-stabilize}
\begin{algorithmic}[1]
\STATE Let $T_t$ be the sink-component at the beginning of round $t$. 
\STATE Initialise $S_0 := \mbox{an arbitrary spanning tree}$, 
       $\mathbb{K}_0:=V$ and 
       $\mathbb{O}_0:=\emptyset$. 
\FOR {round $t := 1$ to $n$}
\STATE Apply FindStable($T_t,S_{t-1},\mathbb{O}_{t-1}$) to compute a
       spanning tree $S_t$.
\STATE Update $\mathbb{O}_t:=\mathbb{O}_{t-1}\cup\calO_t$ and 
       $\mathbb{K}_t:=V\setminus\mathbb{O}_t$. 
\STATE Activate BFS$(\mathbb{O}_t, S_t)$.
\STATE Activate reverse-BFS$(\mathbb{K}_t, S_t)$.
\STATE Pick a node $v^*$ that is the first node activated by reverse-BFS$(\mathbb{K}_t, S_t)$.
\STATE Replace the arc $(v^*,y)$ in $S_t$ by $(v^*,w)$ the arc chosen by $v^*$ in the routing graph.
\STATE Update $\mathbb{O}_t:=\mathbb{O}_t\cup\{v^*\}$. 
\ENDFOR
\end{algorithmic}
\end{algorithm}

Observe that Procedure~\ref{algo:fair-stabilize} is clearly fair because $\mathbb{K}_t$ and
$\mathbb{O}_t$ partition the set of nodes.
%
%
The basic intuition behind the method is that if we can make our routing graph $T_{t+1}$ look like the
spanning tree $S_t$, then every packet will route.
Typically, any activation sequence that attempts to do this, though,
will induce inconsistencies.
This, in turn, will force nodes to go opaque.
But, it turns out that we can make those nodes in $\mathbb{O}_t$ choose
arcs in accordance with $S_t$ via the use of BFS$(\mathbb{O}_t, S_t)$.
(It is not at all obvious that this can be done because nodes in
$\mathbb{O}_t$ may actually be clear, that is, they need not be in
$\calO_t$.) 

Then the question becomes how do we keep track of the packets.
The key point is that nodes in $\mathbb{O}_t$ choose arcs in
accordance with the spanning tree $S_t$.
Therefore, since $\mathbb{O}_1 \subset \mathbb{O}_2 \subset \mathbb{O}_3 \subset
\cdots$ and the containments are strict, we will eventually have
$\mathbb{O}_t=V$, and our routing graph will be a spanning tree.
Thus, every packet routes!
Moreover, the strong stability property of $S_t$ on $\mathbb{O}_t=V$
also implies that the final routing graph is a stable spanning tree.

The following lemma presents the key properties we need to prove all this. 

\begin{lemma}\label{lmm:main-props}
At the end of round $t$ of Fair-Stabilise(), we have that:\\
$\bullet$ The spanning tree $S_t$ has the strong stability on $\mathbb{O}_t$,
and\\
$\bullet$  $\mathcal{A}^+_{S_t}(\mathbb{O}_t)$ is contained in the routing graph at the end of round. Specifically, for each 
node $v\in\mathbb{O}_t$, if $(v,w)\in S_t$ then, upon activation, $v$ chooses $w$ as its next hop. 
\end{lemma}

\begin{proof}
We proceed by induction on $t$. The statement is clearly true for $t=0$
because $\mathbb{O}_0=\emptyset$.
Now, suppose that the statements hold up to round $t-1$ for some $t>0$. 
We first show that $S_t$ has a strong stability property on
$\mathbb{O}_t$ (before adding $v^*$).
To do this, we have to show that FindStable($T_t,S_{t-1},\mathbb{O}_{t-1}$)
is called with a valid input, i.e., 
$S_{t-1}$ is a skeleton of $T_{t}$ and has the strong stability property on $\mathbb{O}_{t-1}$.
The latter fact, that $S_{t-1}$ has the strong stability property on $\mathbb{O}_{t-1}$, follows from the 
induction hypothesis. 

So, we need to show $S_{t-1}$ is a skeleton of $T_{t}$.
By induction, at the end of the previous round (before the learning phase)
we have that $\mathcal{A}^+_{S_{t-1}}(\mathbb{O}_{t-1})$ is
contained in the routing graph.  
So, for each (maximal) subtree 
$F\subseteq S_{t-1}[\mathbb{O}_{t-1}]$ rooted at a node $v$, 
every node of $F$ has a path to the sink $r$ if and only if $v$ has a
path to $r$. 
Thus, either $\mathcal{A}^+_{S_{t-1}}(F)$ is contained in the sink
component $T_{t}$ or no node of $F$ is in the sink component $T_{t}$. 
In other words, $S_{t-1}$ is a skeleton of $T_{t}$, as desired.
Consequently, by Lemma~\ref{lem:ext-stability}, FindStable($T_t,S_{t-1},\mathbb{O}_{t-1}$)
builds a spanning tree $S_t$ that is 
strongly stable on $\mathbb{O}_t$ (before the arc emanating from $v^*$ is updated in Step 8).

Next, we claim that after applying BFS$(\mathbb{O}_t, S_t)$, 
the routing graph looks exactly like $S_t$. 
This is true for all nodes outside $\mathbb{O}_t$ by the construction
of $S_t$.
To see this, observe that the tree $S_t$ is constructed from 
FindStable($T_t,S_{t-1},\mathbb{O}_{t-1}$) by adding arcs joining
opaque nodes in $\calO_t\subseteq\mathbb{O}_t$ to $V(T_t)$. 
Since no nodes in $V\setminus\mathbb{O}_t$ have been activated at this
point, we must have that
$T_t[V\setminus\mathbb{O}_t]\subseteq 
 T_t[V\setminus\calO_t] = S_t[V\setminus\calO_t]$.
For the remaining nodes, we proceed by induction on the order in which nodes
are activated by BFS$(\mathbb{O}_t, S_t)$. 
Our induction hypothesis is that every node activated during this time
will choose its parent in $S_t$ as its next hop and become
consistent, so it has a ``real'' path as its chosen route. 

Let us prove the base case. 
Consider the first node $x$ activated by
BFS$(\mathbb{O}_t, S_t)$. Suppose $x$ chooses $z$ in the routing graph,
that is, $(x,z)\in T_{t}$. 
At the time we activate $v^*$, we have 
(1)~$T_t[V\setminus\calO_t] = S_t[V\setminus\calO_t]$ and
(2)~$z$ is clear and is in $V\setminus\mathbb{O}_t$ (because the sink
node $r$ is in $V\setminus\mathbb{O}_t$).
There are two cases. 
First, if $x$ is opaque, then $z$ is valid for $x$ because 
every node now has a real path as its chosen route. 
Second, if $x$ is clear, then $z$ cannot have $x$ in its chosen route
because $(x,z)$ is in the routing graph. 
Also, each node in the $\mathbb{O}_t$-subtree of $x$ in
$S_t$ either is a descendant of $x$ or is opaque
by (1). Thus, by the strong stability property of $S_t$ on $\calO_t$, 
$z$ is the best valid choice of $x$.
Consequently, $x$ must choose $z$ and must be consistent upon
activation because $z$ has a real path as its chosen route. 

Inductively, assume that every node activated before $x$ by 
BFS$(\mathbb{O}_t,S_t)$ chooses its parent in $S_t$ as its next hop 
and is consistent at the time we activate $x$.  
Now activate $x$ and let $Q_0$ be the set of nodes
already activated by BFS$(\mathbb{O}_t, S_t)$. 
Then $R[Q]=S_t[Q]$ where $R$ is the current routing graph and 
$Q=(V\setminus\calO)\cup Q_0$.
Thus, by induction, the node $x$ must have a parent $z\in Q$ which is clear and
consistent. 
But this situation is now similar to case where we activated the first node 
(with $Q$ replacing $V\setminus\mathbb{O}_t$).
Therefore, we conclude that $x$ must choose $z$ and become consistent.
This proves the claim that the routing graph looks like $S_t$. 

Finally, consider the node $v^*$. 
By the choice of $v^*$, all the descendants of $v^*$ in $S_t$ must be in $\mathbb{O}_t$.
Since the routing graph now looks exactly like $S_t$ and every node is
consistent, $v^*$ can choose any node outside its $\mathbb{O}_t$-subtree
in $S_t$.
Consequently, as $w$ is the best valid choice of $v^*$, the modification of 
the tree $S_t$ (replacing $(v^*,y)$ by $(v^*,w)$) results in a tree
$S_t$ that is strongly stable on $\mathbb{O}_t$ after adding $v^*$. 
Also, since $(v^*,w)$ is now in both the routing graph and $S_t$,
both statements hold at the end of the iteration.
\end{proof}

Lemma~\ref{lmm:main-props} produces a guarantee that the cardinality 
of $\mathbb{O}_t$ is increasing.
\begin{lemma} \label{lem:stab-size-inc}
The cardinality of $\mathbb{O}_t$ strictly increases with each round $t$,
until $\mathbb{O}_t=V$. Furthermore if some packet does not route in 
round $t-1$, then $|\mathbb{O}_t|\geq |\mathbb{O}_{t-1}|+3$. 
\end{lemma}

\begin{proof}
By construction, we have $|\mathbb{O}_t|\geq |\mathbb{O}_t|+1$. 
Thus, the first statement holds.
To prove the second statement,
it suffices to show that if there is an opaque node in $\calO_t$ then
there are at least two opaque nodes not in $\mathbb{O}_{t-1}$.
This would imply that we add at least three nodes to $\mathbb{O}_t$, 
two nodes from $\calO_t\setminus\mathbb{O}_{t-1}$ and 
the node $v^*$ at the end of the iteration.
So, we have to prove that $|\calO_t\setminus\mathbb{O}_{t-1}|\geq 2$
for all $t\geq 1$. 

Suppose $\calO_t\neq\emptyset$; if not then every packet routed in round $t-1$. Now
consider the routing graph $R_{t-1}$ at the end of round $t-1$ (before the
learning phase). Since there is an opaque node in $\calO_t$, 
the routing graph $R_{t-1}$ must contain a cycle $C$. 
By Lemma~\ref{lmm:main-props},
$\mathcal{A}^+_{S_{t-1}}(\mathbb{O}_{t-1})$ is contained in $R_{t-1}$. 
The subgraph $\mathcal{A}^+_{S_{t-1}}(\mathbb{O}_{t-1})$ is a forest
because
$\mathcal{A}^+_{R_{t-1}}[\mathbb{O}_{t-1}]=\mathcal{A}^+_{S_{t-1}}[\mathbb{O}_{t-1}]$ 
and $S_{t-1}$ is a tree. 
So, at least one vertex of $C$ is not in $\mathbb{O}_{t-1}$.
Thus, we have shown that $|\calO_t\setminus\mathbb{O}_{t-1}|\geq 1$.
If $C$ has no vertex in $\mathbb{O}_{t-1}$, then 
$C$ must have at least two vertices and all of them are in 
$\calO_t\setminus\mathbb{O}_{t-1}$, so we are done.
%
If $C$ has some vertex in $\mathbb{O}_{t-1}$, then $C$ must contain an
arc $(u,v)$ such that $u\in\mathbb{O}_{t-1}$ and 
$v\notin\mathbb{O}_{t-1}$. 
Follow the cycle $C$ starting from the vertex $v$. 
If the node $w$ after $v$ in $C$ is not in $\mathbb{O}_{t-1}$,
then we have found two nodes in $\calO_t\setminus\mathbb{O}_{t-1}$, and we are done. 
So, $w$ is in $\mathbb{O}_{t-1}$. 
Now, continue traversing $C$ from $w$ to $u$.
%
By Lemma~\ref{lmm:main-props}, after applying BFS$(\mathbb{O}_{t-1},S_{t-1})$, 
the forest $\mathcal{A}^+_{S_{t-1}}(\mathbb{O}_{t-1})$ is contained in
the routing graph.
In fact, the routing graph at this point is exactly $S_{t-1}$, and no
nodes change their choice.
Hence, all nodes are consistent after applying BFS$(\mathbb{O}_{t-1},S_{t-1})$.
%
This means that $w$ is not a descendant of $u$ or $v$ in 
the forest $\mathcal{A}^+_{S_{t-1}}(\mathbb{O}_{t-1})$;
otherwise, $w$ would have $u$ and $v$ (in fact, the arc $(u,v)$) in
its chosen route. 
So, the only way $C$ can go from $w$ to $u$ is to leave 
the set of nodes $\mathbb{O}_{t-1}$. 
But then the next node $y$ on $C$ we reach outside 
$\mathbb{O}_{t-1}$
satisfies $y\neq v$. Therefore, again, $C$ has two distinct nodes not in $\mathbb{O}_{t-1}$.
\end{proof}

It is immediate from Lemma~\ref{lem:stab-size-inc} that we can route
every packet in $\lfloor n/3 \rfloor$ rounds, and that the network
becomes stable in $n$ rounds. 
Moreover, we can deduce a stronger failure guarantee.
We say that round $t$ is a {\em imperfect round} if we cannot route
every packet. Then
there can be at most $\lfloor n/3 \rfloor$ imperfect rounds (note that these may not be
consecutive rounds) even if the routing graph is not yet stable. 

\begin{theorem}\label{thm:stable-n-route-n3}
There is an activation sequence that
routes every packet in $\lfloor n/3\rfloor$ rounds,
gives a stable spanning tree in $n$ rounds, 
and guarantees that there are at most $\lfloor n/3\rfloor$ imperfect rounds. 
\end{theorem}

\medskip

\noindent {\bf Acknowledgements.}
We thank Michael Schapira and Sharon Goldberg for interesting
discussions on this topic.

\bibliographystyle{plain}
\bibliography{bgp}




\section*{Appendix: Interdomain Routing and Model Technicalities.}
\label{appendix:model}

 The Internet is a union of subnetworks called {\em domains} or
 Autonomous Systems (ASes).
 The inter-domain routing protocol used in the Internet today is
 called the Border Gateway Protocol (BGP), and it works as follows~\cite{Ste98}.
 For destination $r$ and router $v$, each neighbouring router
 of $v$ announces to $v$ the route to $r$ that it has chosen 
 and from amongst these announced routes, $v$ chooses the route $\calP(v)$ that 
 it ranks highest.
 The router $v$ then in turn announces to its neighbouring routers its routing path $\calP(v)$.
 This process continues until an {\em equilibrium} is reached in which
 each router has chosen a route and for each
 router $v$, no neighbour of $v$ announces a route that $v$
 would rank higher than its currently routing path.
 The ranking of routes at a router depends on a 
 number of route attributes
 such as which neighbour announced the route, how long the route is, and which
 domains the route traverses.
 In fact, the ranking of routes at $v$ is a function of $v$'s 
 traffic engineering goals as well as the Service Level Agreements (SLAs), 
 that is, the economic contracts $v$ has made with its neighbours.

 It is well known that BGP can be thought of as a game~\cite{GSW02} and that
 BGP as a game may have no Nash equilibrium~\cite{VGE00,GW99}.
 There is now a vast literature studying the conditions under which BGP will or
 will not have an equilibrium (for example~\cite{GSW02,GR01}).
 It has been shown that in a BGP instance, the absence of a structure known as a
 {\em dispute wheel} implies that the BGP instance will have a unique
 equilibrium~\cite{GSW02}.
 There have been a number of papers analysing the worst-case convergence time 
 of BGP instances having no dispute wheel~\cite{Kar04,FSR11,SZR10}.
 There have also been many experimental papers measuring BGP
 convergence times~\cite{LA+01,GP01}
 and papers offering modifications to BGP with the goal of
 speeding up convergence~\cite{PA+05,KK+07}.

 However, BGP convergence is only a step towards the ultimate goal of successfully
 delivering packets to the destination.
 In fact, routers perform operations simultaneously
 on two basic levels: (1) on the {\em control plane} (i.e., where BGP
 exchanges routing information with other routers as described above)
 and (2) on the {\em forwarding plane} where routers use the routing information
 from BGP to forward packets to
 neighbouring routers towards the packets' ultimate destinations.
 That is, packets are being forwarded during the time that the control
 plane is attempting to settle on an equilibrium.

Recall our model in Section~\ref{prelim}.
We base our idealised routing protocols on BGP and two particularly 
important attributes that a routers uses to rank its available
routes. 
Firstly, a router might not trust certain domains to handle its
packets securely or in a timely fashion, so it may reject routes
traversing such unreliable domains. 
%
This motivates a {\em (no-go) filtering}, which will filter out any route that goes through an undesirable domain (i.e., a domain on the router's no-go filtering list).
Secondly, it has been argued that perhaps the most important attribute in
how a router $v$ ranks routes is the neighbour of $v$ 
announcing the route to $v$~\cite{SZR10}. 
That is, one can think of each router ordering its neighbours and
ranking any route from a lower ordered router over any route from
a higher ordered router.
This is called {\em next-hop routing}.
Thus, in our protocols,  a node ranks routes by first filtering out any route that
goes through nodes on its filtering list and then choosing from amongst the remaining
routes the one announced by the lowest ordered neighbour ({\em next-hop preference with filtering}).

As discussed, to analyse stability, it suffices to consider only the control plane.
But, to understand packet routing, we need to understand the interaction between the forwarding and control planes. 
Thus, we need to incorporate the actions of the forwarding plane into the standard model of the control plane \cite{GSW02}.
To do so, some assumptions must be made, particularly concerning the synchronisation between the planes.
In setting up a model for a practical problem, it is important to examine 
how the modelling assumptions relate to reality. So, here we briefly address
some technical aspects:

\begin{itemize}

\item {\bf Synchronisation of the Planes.}
Observe that, in our model, the control plane and the forwarding plane operate at a similar speed.  
This assumption is the worst case in that it maximises the rate at which 
inconsistencies are produced between the nodes routing paths.
In practice, updates in the control plane are much slower than the rate of packet transfer.   

\item {\bf Packet Cycling.}
When a packet gets stuck in a cycle, we will assume that, at the start of the next round, 
an adversary can position the packet at whichever node in the cycle they wish.

\item {\bf Fair Activation Sequences.}
We insist that activation sequences in the control plane are {\em fair}
in that all nodes update their routes at a similar rate. Clearly, the use of permutations 
ensures fairness.
From the theoretical point of view, fairness is important as it avoids
artificially routing packets by the use of unnatural and pathological activation sequences. For example,
it prohibits the use of activation sequences that are biased towards 
nodes in regions where disconnectivities arise and attempts to fix this by ``freezing''
other nodes until consistency is obtained.
Moreover, in practice, routers timings on the control plane are similar. 

\item {\bf Routing in Rounds.}
The use of rounds (defined by permutations) for routing is not vital
and is used for clarity of exposition and to emphasise fairness.
Also, packet forwarding is clearly not delayed until the end of a ``round"
in practice but, again, this is also not needed for the model.
The assumption is made as it clarifies the arguments needed in   
the analyses.  
For example, forwarding at the end of a round can be shown to be equivalent
to forwarding continuously throughout the round with the planes in sync; that is, 
packets are forwarded immediately and, within a round, the routing path at a node 
is updated just before the first packet a node sees is about to leave it.  
  
\item {\bf Route-Verification.}
Route-verification at the end of the round is our one non-worst case assumption 
and is not a standard aspect of BGP, albeit one that can be incorporated in
a fairly simple fashion by tools such as traceroute or an AS-level
traceroute tool such as that described by Mao~et~al.~\cite{MR+03}.
%
%
Route-verification is the focus of the influential paper of John et al. \cite{JK+08} on consensus 
routing. It is also used in the theory literature on incentives under BGP \cite{LSZ11}.
Due to the manipulative power provided by unfair activation sequences,
it is not hard to simplify our algorithms and omit the route-verification step
given the use of unfair activation sequences; see also \cite{SZR10}.
It remains an interesting open problem to obtain consistency using fair sequences without
route-verification.

\item {\bf Filtering.}
In this paper, we assume that each node can apply what is known                                                                     
as {\em import filtering} -- that is, not accepting certain routes from its neighbours.                                                            
This implicitly assumes that each node announces its routing path to all of its neighbours.                                            
In reality, each node may choose to apply {\em export filtering} -- that is, it may announce                                              
any particular route to only a subset of its neighbours (e.g., in order                                                         
to assure ``valley-free routing'' \cite{EH+07}).                                                                 

Export filtering can be incorporated into our model by allowing for neighbour specific import 
filtering rules, where a node $v$ can have a filtering list $\calD(v,w)$ for each neighbour $w$.                                                         
Of course, our lower bounds would still hold for this more general model, but it                                                        
would allow for more special cases to explore.                    

\end{itemize}


\end{document}